\RequirePackage[l2tabu,orthodox]{nag}
\documentclass
[11pt,letterpaper]
{article} 

\usepackage[notes=true,later=false,camera=false]{dtrt}
\usepackage[utf8]{inputenc}
\usepackage{etex}
\usepackage{thmtools}
\usepackage{ stmaryrd }
\usepackage{xspace,enumerate}
\usepackage[T1]{fontenc}
\usepackage[full]{textcomp}
\usepackage[american]{babel}
\usepackage{mathtools}
\usepackage{amsthm}

\usepackage{empheq}

\hypersetup{
colorlinks=true,
urlcolor=Cerulean,
linkcolor=RoyalBlue,
citecolor=OliveGreen,
linktocpage=true,
}
\renewcommand*\backref[1]{\ifx#1\relax \else (pg. #1) \fi}

\usepackage[capitalise,nameinlink]{cleveref}
\crefname{lemma}{Lemma}{Lemmas}
\crefname{fact}{Fact}{Facts}

\crefname{theorem}{Theorem}{Theorems}
\crefname{corollary}{Corollary}{Corollaries}
\crefname{claim}{Claim}{Claims}
\crefname{example}{Example}{Examples}
\crefname{algorithm}{Algorithm}{Algorithms}
\crefname{problem}{Problem}{Problems}
\crefname{definition}{Definition}{Definitions}
\crefname{equation}{Eq.}{Eq.}
\crefname{strategy}{Strategy}{Strategies}

\usepackage{paralist}
\usepackage{turnstile}
\usepackage[framemethod=TikZ]{mdframed}
\mdfsetup{frametitlealignment=\center}
\usepackage{tikz}
\usepackage{caption}
\DeclareCaptionType{Algorithm}
\usepackage{newfloat}
\usepackage{aliascnt}
\newtheorem{theorem}{Theorem}[section]
\newaliascnt{lemma}{theorem}
\newtheorem{lemma}[lemma]{Lemma}
\aliascntresetthe{lemma}\newtheorem*{lemma*}{Lemma}
\newaliascnt{remark}{theorem}
\newtheorem{remark}[remark]{Remark}
\aliascntresetthe{remark}\newtheorem*{remark*}{Remark}
\newaliascnt{definition}{theorem}
\newtheorem{definition}[definition]{Definition}
\aliascntresetthe{definition}\newtheorem*{definition*}{Definition}

\newaliascnt{proposition}{theorem}

\aliascntresetthe{proposition}
\newaliascnt{fact}{theorem}
\newtheorem{fact}[fact]{Fact}
\aliascntresetthe{fact}
\newaliascnt{corollary}{theorem}
\newtheorem{corollary}[corollary]{Corollary}
\aliascntresetthe{corollary}

\newtheorem{algorithm-thm}[theorem]{Algorithm}

\usepackage[
letterpaper,
top=1.2in,
bottom=1.2in,
left=1in,
right=1in]{geometry}
\usepackage{newpxtext} 
\usepackage{textcomp} 
\usepackage{mathpazo}
\usepackage[scr=rsfso]{mathalfa}
\usepackage{bm} 
\let\mathbb\varmathbb
\usepackage{microtype}

\usepackage{footnotebackref}

\allowdisplaybreaks
\newcommand{\FormatAuthor}[3]{
\begin{tabular}{c}
#1 \\ {\small\texttt{#2}} \\ {\small #3}
\end{tabular}
}

\newcommand{\eps}{\varepsilon}

\newcommand{\E}{{\mathbb E}}

\newcommand{\ip}[1]{\langle #1 \rangle}

\newcommand{\cH}{\mathcal H}

\newcommand{\mper}{\,.}
\newcommand{\mcom}{\,,}

\newcommand{\cC}{\mathcal C}

\renewcommand{\geq}{\geqslant}
\renewcommand{\leq}{\leqslant}
\renewcommand{\epsilon}{\varepsilon}

\usepackage{macros}
\title{Lower Bounds for all List-Decodable Deletion Codes}
\author{
 \begin{tabular}{cc}
\FormatAuthor{Andrew D. Lin}{andrewlin@princeton.edu}{Princeton University}
  \end{tabular}
  }
  \date{\today}
\begin{document}

\newcommand{\fhj}{F_\cH(j)}
\newcommand{\fgj}{F_{\gnkt}(j)}
\newcommand{\lcs}{\text{LCS}}
\newcommand{\scs}{\text{SCS}}
\newcommand{\mlcs}[1]{m_{\lcs(u,v),#1}}
\newcommand{\mscs}[1]{m_{\scs(u,v),#1}}

\maketitle
\begin{abstract}
    A length-$n$ binary $k$-deletion code is a set of binary strings such that if we delete any $k$ bits of a string, leaving a length-$(n-k)$ binary string, we can uniquely recover the codeword. In this paper, we consider $t$-list decodable deletion codes, where after $k$ bits of a codeword are deleted, we can identify a list of size at most $t$ such that the original codeword lies in the list. We prove a lower bound of $\Omega_k(2^n t\log^{1/t}n/n^{k+k/t})$ on the optimal size of a $t$-list decodable $k$-deletion code, giving a $\sqrt{\log n}$ improvement over the previously best known bounds for $2$-list decodable $2$-deletion codes \cite{GuruswamiH21} and providing the first nontrivial lower bound when $t>2$ or $k>2$. Our bound holds for all $t\leq n^k$, showing that $t=\Omega(\log n)-$list decodable deletion codes have optimal size $\Theta_k(2^n t/n^k)$, asymptotically matching the known upper bound. We also prove upper bounds on the number of common subsequences and common supersequences of a given length for any two binary strings.
\end{abstract}
\section{Introduction}
The $k$-deletion channel for binary strings is a noise model which takes in a length-$n$ binary string $x\in \{0,1\}^n$ and outputs an arbitrary length-$(n-k)$ subsequence $y\in \{0,1\}^{n-k}$ of the input. The receiver has no information other than the subsequence; in particular, the positions of the $k$ deleted bits are unknown. A $k$-deletion code is an error correcting code for the $k$-deletion channel, where we can decode the original codeword given a length-$(n-k)$ subsequence. More precisely, it is a subset $C\subseteq \{0,1\}^n$ such that no two elements of $C$ share a common length-$(n-k)$ subsequence. A $t$-list decodable code is a relaxation of an error correcting code, where instead of recovering the codeword from a corrupted input, we allow ourselves to find a list of at most $t$ codewords that the input can be. Determining the optimal size of various types of codes is a key question in coding theory.

Letting $D(n,k)$ denote the maximum size of a length-$n$ $k$-deletion code, Levenshtein \cite{Levenshtein65} showed the upper and lower bounds of
\begin{align*}
    \Omega_k\left(\frac{2^n}{n^{2k}}\right)\leq D(n,k)\leq O_k\left(\frac{2^n}{n^k}\right)\mper
\end{align*}
Heuristically, the upper bound follows from the fact that most length-$n$ binary strings have $\Omega_k(n^k)$ length-$(n-k)$ subsequences, none of which can be a subsequence of another codeword. The lower bound follows from greedily picking codewords, using the fact that each of the $O_k(n^k)$ length-$(n-k)$ subsequences of a chosen string each have $O_k(n^k)$ length-$n$ supersequences which cannot be chosen. For $k=1$, the Varshamov and Tenengolts \cite{VarshamovTenengolts65} construction was shown by \cite{Levenshtein65} to be an explicit $1$-deletion code of size $\Theta(2^n/n)$, matching the upper bound. Precisely, \cite{VarshamovTenengolts65} defines for all $0\leq a\leq n$ the set $\cC_a\coloneqq \{x\in \{0,1\}^n:\sum_{i}ix_i=a\mod n+1\}$; the largest $\cC_a$ is a $1$-deletion code of size at least $2^n/(n+1)$ \footnote{In fact, \cite{Ginzburg1967} shows that $|\cC_a|=\Theta(2^n/n)$ for all $a$.}. For $k\geq 2$, a recent work \cite{ABGHK} improved the lower bound by a logarithmic factor, showing that $D(n,k)\gtrsim 2^n\log n/n^{2k}$; however, closing the existential $\tilde{O}(n^k)$ gap between the upper and lower bounds for $k\geq 2$ remains open, with little progress made over decades.

We thus consider list-decodable deletion codes. A work of Guruswami and Hastad \cite{GuruswamiH21} which explicitly constructed $2$-deletion codes of size $2^n/n^{4+o(1)}$, nearly matching the existential bound up to a $n^{o(1)}$ factor, also showed the existence of $2$-list decodable $2$-deletion codes of size $\Omega(2^n/n^3)$, giving the first improvement on the size of $t$-list decodable codes above the $\Omega(2^n\log n/n^{2k})$ lower bound when $t=1$. However, no improved lower bound is known for any $k\geq 3$ and for $k=2$, the $\Omega(2^n/n^3)$ lower bound remains the best known for any $t$. We note that the upper bound of \cite{Levenshtein65} translates to an $O_k(2^n t/n^k)$ upper bound for $t$-list decodable codes. 

\parhead{Our results}In this work, we prove the first nontrivial lower bound for all $t,k$, which can be viewed as a generalization of \cite{ABGHK} to binary list-decodable deletion codes. The approach of \cite{ABGHK} considers the $k$-deletion graph on vertex set $\{0,1\}^n$, where two vertices are connected if they share a common length-$(n-k)$ subsequence; a $k$-deletion code is simply an independent set in this graph. They then show the existence of a size $\Omega_k(2^n\log n/n^{2k})$ independent set by first proving the graph contains fewer than $2^n d_{\text{avg}}^{2-\eps}$ triangles, where $d_{\text{avg}}$ is its average degree, and using a classic result \cite{AKS80,AKS81} to obtain a $\log d_{\text{avg}}$ improvement compared to the lower bound given by Turan's theorem. 

For list-decodable deletion codes, we generalize the notion of a $k$-deletion graph, where we have a hyperedge for any $(t+1)$ vertices which all share a common length-$(n-k)$ subsequence. We call this the $(k,t)$-deletion hypergraph. Similar lower bounds on the independence number of uniform hypergraphs exist: if a $(t+1)$-uniform hypergraph $\cH$ on $N$ vertices with maximum degree $\Delta$ is \emph{uncrowded}, meaning it has no cycles\footnote{A cycle in a hypergraph is defined as a closed path of distinct $v_1, \cdots, v_\ell,v_{\ell+1}=v_1$ and distinct hyperedges $C_1, \cdots, C_\ell$ such that $v_i,v_{i+1}\in C_i$ for all $1\leq i\leq\ell$.} of length $\leq 4$, then $\alpha(\cH)\gtrsim N(\log \Delta/\Delta)^{1/t}$ \cite{KPS82, AKPSS}. A follow-up result of \cite{DLR95} showed the aforementioned bound for all hypergraphs such that the number of pairs of hyperedges intersecting in size $j$ is not too large for any $j\geq 2$. We strengthen this result and use it to prove the bound on the independence number of the $(k,t)$-deletion hypergraph, which implies an existential lower bound on the size of a $t$-list decodable $k$-deletion code, which we formally state below.
\begin{theorem}\label{thm:main}
    For all $t\leq n^k$, there exists a binary $t$-list decodable $k$-deletion code $\cC\subseteq\{0,1\}^n$ with size $|\cC|\geq \Omega_k\left(\frac{2^n t\log^{1/t} n}{n^{k(1+1/t)}}\right)$.
    
\end{theorem}
As a direct corollary, we obtain tight bounds on the optimal size of list-decodable deletion codes with list size $\Omega(k\log n)$.
\begin{corollary}
    The optimal size of a $t$-list decodable $k$-deletion code is $\Theta_k(2^nt/n^k)$ for all $\Omega(k\log n)\leq t\leq n^k$.
\end{corollary}

Our proof also requires upper bounds on the number of common supersequences of a fixed length of two strings. We first define some notation.
\begin{definition}
    Let $u$ and $v$ be strings. We call a common supersequence $w$ of $u$ and $v$ a \emph{minimal common supersequence} of $u$ and $v$ if no proper subsequence of $w$ is a common supersequence of $u$ and $v$. Let $\lcs(u,v)$ denote the length of the longest common subsequence of $u$ and $v$ and $\scs(u,v)$ denote the length of the shortest common supersequence of $u$ and $v$. Let $m_{\lcs(u,v)}$ and $m_{\scs(u,v)}$ denote the number of common subsequences or supersequences of length $\lcs(u,v)$ or $\scs(u,v)$, respectively. For all $r\geq 0$, let $\mlcs{r},\mscs{r}$ denote the number of common subsequences or supersequences of length $\lcs(u,v)-r$ or $\scs(u,v)+r$, respectively. 
\end{definition}
The following theorem generalizes \cite[Theorem~3]{ABGHK}, which shows the $r=0$ case of our result.
\begin{theorem}\label{thm:genmult}
    Let $n,a,b$ be natural numbers such that $n\geq a+b.$ If $u$ and $v$ are words of length $n-a$ and $n-b$, respectively, and $\lcs(u,v)=n-a-b$, then for all $r\geq 0$, $\mlcs{r}\leq\mscs{r}$. Furthermore, if $u$ and $v$ are binary strings, then
    \begin{align*}
         \mlcs{r}\leq\mscs{r}\leq \sum_{s=0}^r {a+b+2s\choose a+s}\sum_{i=0}^{r-s}{n+r\choose i}\mper
    \end{align*}
\end{theorem}

\section{Preliminaries}
We formally define a $t$-list decodable $k$-deletion code and the $(k,t)$-deletion hypergraph.
\begin{definition}
    We call $\cC\subseteq \{0,1\}^n$ a \emph{$t$-list decodable $k$-deletion code} if for all $y\in \{0,1\}^{n-k}$, there exist at most $t$ distinct $x\in\cC$ containing $y$ as a subsequence.
\end{definition}
\newcommand{\gnkt}{\Gamma_{n,k,t}}
\begin{definition}
    We define the $(k,t)$-deletion hypergraph $\gnkt$ to be the $(t+1)$-uniform hypergraph with vertex set $V(\gnkt)\coloneqq \{0,1\}^n$ and edge set $E(\gnkt)$ such that for all $C\subseteq \{0,1\}^n$ with $|C|=t+1$, we have $C\in E(\gnkt)$ if there exists some $y\in \{0,1\}^{n-k}$ which is a subsequence of all $x\in C$.
\end{definition}
We use the following result which lower bounds the independence number of any uniform hypergraph such that the number of pairs of hyperedges intersecting in exactly $j$ coordinates is not too large for all $j\geq 2$.
\begin{definition}
    Let $\cH$ be a $k$-uniform hypergraph. Define $\fhj$ to be the number of unordered pairs $C,C'\in \cH$ such that $|C\cap C'|=j$.
\end{definition}
\begin{fact}(\cite[Theorem~3]{DLR95})\label{fact:indnum}
    Let $3\leq k\leq O((\log t/\log\log t)^{1/3})$ and $\cH$ be a $k$-uniform hypergraph with $n$ vertices with maximum degree $\leq t^{k-1}$. If $\fhj\leq nt^{2k-j-1-\eps}$ for some $\eps\geq \exp(-O(k))$ for all $2\leq j\leq k-1$, then $\alpha(\cH)=\Omega(\frac{n}{t}\log^{1/(k-1)}t)$.
\end{fact}
\begin{remark}
    \cite{DLR95} states an $\alpha(\cH)=\Omega_{k,\eps}(\frac{n}{t}\log^{1/(k-1)}t)$ lower bound; the dependence on $\eps,k$ is a factor of $(\eps/k)^{1/(k-1)}$ which is $\Omega(1)$ for $\eps\geq \exp(-O(k))$.
\end{remark}

\section{Proofs of \Cref{thm:main} and \Cref{thm:genmult}}
We begin by proving the following bound on the number of minimal common supersequences of a given length.
\begin{lemma}\label{lem:mincount}
    Let $n, a, b$ be natural numbers with $n \geq a + b$ and let $r\geq 0$. If $u$ and $v$ are words with length $n - a$ and $n - b$ (respectively) and $\scs(u, v)= n$, then the number of minimal common supersequences of $u$ and $v$ of length $n+r$ is at most ${a+b+2r\choose a+r}$.
\end{lemma}
\begin{proof}
    We will prove the result for $r\geq -1$. If $a=0$ or $b=0$, then there is at most one minimal common supersequence, while when $a,b\geq 1$ and $r= -1$, there are no minimal common supersequences of length $n-1$. We induct on $a+b+2r$ with the aforementioned base cases. Assume without loss of generality that $u_1\neq v_1$; if $u$ and $v$ have a common prefix, then every minimal common supersequence of $u$ and $v$ must have the same prefix, so we can ignore the prefix. Then every minimal common supersequence is either of the form $u_1x$ or $v_1y$, where $x$ is a minimal common supersequence of $u_{[2,n-a]}$ and $v$, and $y$ is a minimal common superseqeunce of $u$ and $v_{[2,n-b]}$. Note that $x$ has length $n+r-1$ and $\scs(u_{[2,n-a]},v)$ is either $n-1$ or $n$; if it is equal to $n-1$, then $a$ and $r$ are unchanged but $b$ decreases by $1$, and if it is $n$, then $a$ increases by $1$, $b$ is unchanged, and $r$ decreases by $1$. Either way, there are at most ${a+b+2r-1\choose a+r}$ total choices for $x$ by induction. A similar analysis shows that there are at most ${a+b+2r-1\choose a+r-1}$ choices for $y$. Thus by Pascal's Identity, there are at most ${a+b+2r-1\choose a+r}+{a+b+2r-1\choose a+r-1}={a+b+2r\choose a+r}$ length $n+r$ minimal common supersequences.
\end{proof}
We use this to prove an upper bound on the number of common supersequences of a given length.
\begin{lemma}\label{lem:upperbound}
    Let $n,a,b$ be natural numbers such that $n\geq a+b.$ If $u\in \{0,1\}^{n-a}$ and $v\in \{0,1\}^{n-b}$ such that $\scs(u,v)=n$, then
    \begin{align*}
         \mscs{r}\leq \sum_{s=0}^r {a+b+2s\choose a+s}\sum_{i=0}^{r-s}{n+r\choose i}\mper
    \end{align*}

\end{lemma}
\begin{proof}
    Let $x$ be a length $\scs(u,v)+r=n+r$ common supersequence of $u$ and $v$. Then there exists $y\prec x$ such that $y$ is a minimal common supersequence of $u$ and $v$. Note that we can count the number of possible $x$ by first picking the length of some minimal common supersequence, which must be $n+s$ for some $0\leq s\leq r$. Then by \Cref{lem:mincount} there are up to ${a+b+2s\choose a+s}$ minimal common supersequences $y$ of that length, and finally, we generate a length $n+r$ common supersequence $x\succ y$, which can be done in $\sum_{i=0}^{r-s}{n+r\choose i}$ ways.
\end{proof}

Now we prove \Cref{thm:main}, which follows from bounds on the independence number of the $(k,t)$-deletion hypergraph $\gnkt$.

\begin{lemma}\label{lem:weakalphabound}
 For all $0<t\leq n^k$ and $k\geq 2$, we have $\alpha(\gnkt)\geq\Omega(\frac{2^nt}{n^{k+k/t}})$.    
\end{lemma}
\begin{proof}
    Let $p=t/en^{k+k/t}$ and let $\cC\subseteq \{0,1\}^n$ such that each $x\in \{0,1\}^n$ is included in $\cC$ independently with probability $p$. Then $\E|\cC|=2^nt/en^{k+k/t}$. Note that every $y\in \{0,1\}^{n-k}$ has at most $\sum_{i=0}^k{n\choose i}\leq n^k$ length-$n$ supersequences and thus $|E(\gnkt)|\leq 2^{n-k}{n^k\choose t+1}\leq 2^{n-k}n^{k(t+1)}/(t+1)!$. Given $\cC$, we can delete one $x\in C$ from $\cC$ from every $C\in E(\gnkt)$ where $C\subseteq \cC$ to create an independent set. The expected size of such an independent set is at least
    \begin{align}
        \frac{2^nt}{en^{k+k/t}}-\frac{2^{n-k}n^{k(t+1)}t^{t+1}}{\left(n^{k+k/t}\right)^{(t+1)}e^{t+1}(t+1)!}\geq \frac{2^nt}{en^{k+k/t}}\left(1-\frac{1}{2^{k+1/2}\sqrt{\pi t}}\right)\mper
    \end{align}
Thus we have $\alpha(\gnkt)\geq\Omega_k(\frac{2^nt}{n^{k+k/t}})$.
\end{proof}
\begin{lemma}\label{lem:alphabound}
 For all $2\leq t\leq O((\log n/\log \log n)^{1/3})$, we have $\alpha(\gnkt)\geq\Omega_k(\frac{2^nt}{n^{k+k/t}}\log^{1/t}n)$.
\end{lemma}
\begin{proof}
    Let $N=2^n$ and $T=An^{k+k/t}/t$, where $A=O_k(1)$ will be chosen later. Pick an arbitrary $x\in V(\gnkt)$. Then for every $C\in E(\gnkt)$ such that $x\in C$, there exists some $y\in \{0,1\}^{n-k}$ such that $y\prec x'$ for all $x'\in C$. Then given $x$, there are ${n\choose k}\leq n^k$ ways to choose such a $y\prec x$. There are at most $2^k{n\choose k}\leq 2n^k$ total $x'\in \{0,1\}^n$ with $x'\succ y$, so we upper bound the number of ways to choose the remaining elements of some $C\ni x$ by ${2n^k\choose t}\leq 2^tn^{kt}/t!$. Thus $x$ has degree at most $n^k2^tn^{kt}/t!\leq (2en^{k+k/t}/t)^t\leq T^t$, where we ensure that we choose $A\geq 2e$.

    We now upper bound $\fgj$, after which the result follows from \Cref{fact:indnum}. For any $2\leq j\leq t$, we upper bound the number of ways to count pairs $C,C'\in \gnkt$ such that $|C\cap C'|=j$. To do so, we first pick $y,y'$ such that $y\prec x$ and $y'\prec x'$ for all $x\in C,x'\in C'$, counting only $y,y'$ pairs such that there exists $x\in \{0,1\}^n$ with $y,y'\prec x$. Then we choose $C\cap C'$ and finally picking $C\setminus C'$ and $C'\setminus C$.

    Start by picking $y\in \{0,1\}^{n-k}$, which can be done in $2^{n-k}$ ways. Note that if we have $y', y\in \{0,1\}^{n-k}$ such that there exists some $x\in \{0,1\}^{n}$ satisfying $y,y'\prec x$, then $n-k\leq\scs(y,y')\leq n$. We pick $y'$ by first picking $n-k\leq\ell\leq n$, then picking a length-$\ell$ supersequence, which can be done in at most $2^{\ell-(n-k)}{\ell\choose \ell-(n-k)}=O_k(n^{\ell-(n-k)})$ ways, and finally some $y'$ satisfying $\scs(y,y')=\ell$ as a length-$(n-k)$ subsequence of the aforementioned supersequence. This can be done in up to ${\ell\choose \ell-(n-k)}=O_k(n^{\ell-(n-k)})$ ways. Then by \Cref{lem:upperbound}, the total number of length-$n$ common supersequences of $y,y'$ is at most 
    \begin{align*}
        m_{\scs(y,y'),n-\ell}\leq\sum_{s=0}^{n-\ell}{2k+2s\choose k+s}\sum_{i=0}^{n-\ell-s}{n\choose i}\leq O_k(n^{n-\ell})\mcom
    \end{align*}
    so there are at most ${O_k(n^{n-\ell})\choose j}\leq \frac{B^j}{j!}\cdot O_k(n^{j(n-\ell)})$ ways to choose $C\cap C'=\{x_1, \cdots, x_j\}$ such that $y,y'\prec x_1, \cdots, x_j$, where $B=O_k(1)$. Therefore, the total number of ways to choose $y,y',x_1, \cdots, x_j$ is
    \begin{align*}
        \frac{B^j}{j!}2^{n-k}\sum_{\ell=n-k}^{n}O_k(n^{j(n-\ell)+2(\ell-(n-k))})
        =\frac{B^j}{j!}2^{n-k}\sum_{\ell=n-k}^nO_k(n^{jk-(j-2)(\ell+k-n)})=\frac{B^j}{j!}2^{n-k}O_k(n^{jk})\mper
    \end{align*}    
    Finally, we upper bound the number of ways to pick $C\setminus C'$ and $C'\setminus C$ by the number of ways to choose $t+1-j$ distinct $x,x'\in \{0,1\}^n$ with $y\prec x$ and $y'\prec x'$, respectively. There are at most $\sum_{i=0}^k{n\choose i}=O_k(n^k)$ choices for each, so there are at most ${O_k(n^k)\choose t+1-j}^2\leq \frac{D^{2(t+1-j)}}{((t+1-j)!)^2}O_k(n^{2k(t+1-j)})$ total ways to choose the remaining elements of $C$ and $C'$, where $D=O_k(1)$. Therefore, letting $\eps = \frac{1}{t+1}$ and picking $A=9e^2(B+D+O_k(1))^2=O_k(1)$, we have
    \begin{align*}
        \fgj
        &\leq 2^{n-k}\frac{B^jD^{2(t+1-j)}}{j!((t+1-j)!)^2}O_k(n^{jk+2k(t+1-j)})
        \\&\leq \frac{(B+D+O_k(1))^{2t+2-j}}{j!((t+1-j)!)^2}\cdot 2^nn^{(k+k/t)(2t-jt/(t+1))}
        \\&\leq \frac{(3eB+3eD+O_k(1))^{2t+2-j}}{t^{2t+2-j}} 2^nn^{(k+k/t)(2t+1-j-1/(t+1))}
        \\&\leq N\left(\frac{A n^{(k+k/t)}}{t}\right)^{2t+1-j-1/(t+1)}
        \leq NT^{2t+1-j-\eps}\mper
    \end{align*}
\end{proof}
\begin{proof}[Proof of \Cref{thm:main}]
The $t=1$ case is known from \cite{ABGHK}. For the $k=1$ case, the \cite{VarshamovTenengolts65} construction gives $n+1$ mutually disjoint $1$-deletion codes of size $\Theta(2^n/n)$, and the union of any $t$ of these is a $t$-list decodable $1$-deletion code of size $\Theta(2^nt/n)\geq \Omega(2^nt\log^{1/t}n/n^{1+1/t})$. When $k,t\geq 2$, since every independent set in $\gnkt$ is a $t$-list decodable $k$-deletion code, so if $t\leq O(\log \log n)$, this directly follows from \Cref{lem:alphabound}. When $t=\Omega(\log \log n)$, we note that $\log^{1/t}n=O(1)$, so the result follows from \Cref{lem:weakalphabound}.
\end{proof}
Finally, we prove the following lemma, which combined with \Cref{lem:upperbound} proves \Cref{thm:genmult}.
\begin{lemma}
    Let $n,a,b$ be natural numbers such that $n\geq a+b.$ If $u$ and $v$ are words of length $n-a$ and $n-b$, respectively, and $\lcs(u,v)=n-a-b$, then for any $r\geq 0$, $\mlcs{r}\leq\mscs{r}$. 
\end{lemma}
\begin{proof}
    We show the existence of an injective mapping $\phi$ from the set of common subsequences of length $\lcs(u,v)-r$ of $u$ and $v$ to the set of common supersequences of length $\scs(u,v)+r$. For any common subsequence $w$ of length $\lcs(u,v)-r$, consider the left-most copy of $w$ in $u$ and $v$, and let $i_1< \cdots< i_{\lcs(u,v)-r}$ be the corresponding coordinates of $u$ and $j_1< \cdots< j_{\lcs(u,v)-r}$ be the corresponding coordinates of $v$. Letting $i_0=j_0=0$, for all $k=1, \cdots, \lcs(u,v)-r$, we print all $u_{[i_{k-1}+1,i_k-1]}$, followed by $v_{[j_{k-1}+1,j_k-1]}$, followed by $w_k$. Finally, we end by printing $u_{[i_{\lcs(u,v)-r}+1,n-a]}$ followed by $v_{[j_{\lcs(u,v)-r}+1,n-b]}$. We let $\phi(w)$ be the resulting length-$(n+r)$ supersequence.  

    Let $x=\phi(w)$ and $x'=\phi(w')$ for some length-$(n-a-b-r)$ common subsequences $w\neq w'$. We will show that $x\neq x'$. Let $i_1<\cdots<i_{\lcs(u,v)-r}$ and $j_1<\cdots<j_{\lcs(u,v)-r}$ be the coordinates of the leftmost copy of $w$ in $u$ and $v$, and $i'_1<\cdots<i'_{\lcs(u,v)-r}$ and $j'_1<\cdots<j'_{\lcs(u,v)-r}$ be the coordinates of the leftmost copy of $w'$ in $u$ and $v$. Let $k$ be the smallest coordinate such that $w_k\neq w'_k$. Then $i_\ell=i'_\ell$ and $j_\ell=j'_\ell$ for all $\ell < k$, and $x$ and $x'$ begin with the same first $i_{k-1}+j_{k-1}-(k-1)$ coordinates. 
    
    If $i_k+j_k=i'_k+j'_k$, then $x_{i_k+j_k-k}=w_k\neq w'_k=x'_{i'_k+j'_k-k}$. Otherwise, we assume without loss of generality that $i_k+j_k<i'_k+j'_k$. Suppose that $x=x'$. If $i_k\geq i'_k$, then $x_{i'_{k}+j'_{k-1}-(k-1)}=u_{i'_k}=w'_k$. Then if $x'_{i'_k+j'_{k}-(k-1)}=w'_k$, this implies that either there exists some $i'_{k-1}+1\leq i\leq i'_k-1$ such that $u_i=w'_k$, or some $j'_{k-1}+1\leq j\leq j'_k-1$ such that $v_j=w'_k$, both of which contradict that $u_{i'_k},v_{j'_k}$ are the coordinates of $w'_k$ in the leftmost copy of $w'$. Thus we have $i_k<i'_k$. Similarly, we have $j_k<j'_k$. 

    Now let $p, q\geq 0$ such that $i_{k+p}\leq i'_k< i_{k+p+1}$ and $j_{k+q}< j'_k\leq j_{k+q+1}$, where we define $i_{\lcs(u,v)-r+1}=n-a+1$ and $j_{\lcs(u,v)-r+1}=n-b$. If $p\leq q$, then $x_{i'_k+j_{k+p}-(k+p)}=u_{i'_k}=w'_k$, and $i'_k+j_{k+p}-(k+p) < i'_k+j'_k-k$. On the other hand, if $p>q$, then $x_{i_{k+q+1}+j'_k-(k+q+1)}=v_{j'_k}=w'_k$, and $i_{k+q+1}+j'_k-(k+q+1)<i'_k+j'_k-k$. Either way, there exists some $i_{k-1}+j_{k-1}-(k-1) < t < i'_k+j'_k-k$ such that $x_t=w'_k$. Now suppose that $x'=x$, so $x'_t=w'_k$. But then as before there exists some $i'_{k-1}+1\leq i\leq i'_k-1$ such that $u_i=w'_k$, or some $j'_{k-1}+1\leq j\leq j'_k-1$ such that $v_j=w'_k$, again contradicting that $u_{i'_k},v_{j'_k}$ are the coordinates of $w_k'$ in the leftmost $w'$ in $u$ and $v$. 
    
    Thus we have $x'\neq x$ for all $w\neq w'$, so $\phi$ is an injective map from length $\lcs(u,v)-r$ subsequences to length $\scs(u,v)+r$ supersequences, so $\mlcs{r}\leq\mscs{r}$. 
    
\end{proof}
\section*{Acknowledgments}
We thank Noga Alon, Hsin-Po Wang, and Wei-Hsuan Yu for useful discussions. GPT-5.5 Plus and GPT-5.6 Pro were used for literature review and proofreading. 
\bibliographystyle{alpha}
\bibliography{bib}
\end{document}